\documentclass[aps,pra,onecolumn,nofootinbib,superscriptaddress]{revtex4}
\usepackage{amsmath,amssymb,amsthm}
\usepackage{amsfonts}
\usepackage{amssymb}
\usepackage{graphicx}
\usepackage{comment}
\usepackage{color}
\usepackage[a4paper, total={6in, 8in}]{geometry}

\newtheorem{theorem}{Theorem}
\newtheorem{definition}[theorem]{Definition}

\newtheorem{lemma}[theorem]{Lemma}
\usepackage[title]{appendix}

\DeclareMathOperator{\Tr}{Tr}

\begin{document}

\title{Mutually Unbiased Unitary Bases in $\mathcal{H}_d$ for $d$-dimensional subspace of $M (d,\mathbb{C})$}

\author{Rinie N. M. Nasir}
\affiliation{Faculty of Science, International Islamic University Malaysia (IIUM),
Jalan Sultan Ahmad Shah, Bandar Indera Mahkota, 25200 Kuantan, Pahang, Malaysia}
\author{Jesni Shamsul Shaari}
\affiliation{Faculty of Science, International Islamic University Malaysia (IIUM),
Jalan Sultan Ahmad Shah, Bandar Indera Mahkota, 25200 Kuantan, Pahang, Malaysia}
\affiliation{ Institute of Mathematical Research (INSPEM), University Putra Malaysia, 43400 UPM Serdang, Selangor, Malaysia.}
\author{Stefano Mancini}
\affiliation{School of Science \& Technology, University of Camerino, I-62032 Camerino, Italy}
\affiliation{ INFN Sezione di Perugia, I-06123 Perugia, Italy}

\date{\today}
\begin{abstract}
Akin to the idea of complete sets of Mutually Unbiased Bases for prime dimensional Hilbert spaces, $\mathcal{H}_d$,  we study its analogue for a $d$ dimensional subspace of $M (d,\mathbb{C})$, i.e. Mutually Unbiased Unitary Bases (MUUBs) comprising of unitary operators. We note an obvious isomorphism between the vector spaces and beyond that, we define a relevant monoid structure for $\mathcal{H}_d$ isomorphic to one for the subspace of $M (d,\mathbb{C})$. This provides us not only with the maximal number of such MUUBs, but also a recipe for its construction.
{\keywords{Foundations of quantum mechanics \and Mutually Unbiased Bases (MUBs) \and Mutually Unbiased Unitary Bases (MUUBs)}}
\end{abstract}

\maketitle



\section{Introduction}
Notwithstanding the rich details of the field of quantum information theory, it can be described in a nutshell as dealing with the notion of retrieving or manipulating information encoded via quantum mechanical properties of a system. In the context of quantum physics, these are related to the fundamental issues of state estimation and unitary transformations. 

Quantum state estimation is probabilistic in nature as measurements of a system in a certain basis different from which it was prepared would generally result in some random outcome. The concept of mutually unbiased bases (MUB), introduced in 1960 by Schwinger \cite{Schwinger} can be seen as describing the extreme scenario where the transition from any state in one basis to any state in the other basis is equiprobable. Various constructions of MUB has been reported in \cite{Alltop,Ivanovic,Wootters,Band}.  We refer the reader to \cite{Durt} for an excellent review on the matter. MUB has even found application in the field of quantum key distribution (QKD) where quantum states from different MUBs act as information carriers \cite{QKD}.

Motivated by the study of MUB, ref.\cite{Jesni} considered the notion of Mutually Unbiased Unitary Bases (MUUB) based on considering the idea of equiprobable guesses of unitary transformations. It is noteworthy that prior to this, ref.\cite{Scott} introduced the notion of MUUB to construct a unitary-2 design related to process tomography of unital quantum channels, albeit a slightly less general definition in comparison.

It was argued in \cite{Jesni}, that, for $d$ being a prime number, MUUBs can only exist for the $d$ and $d^2$ dimensional subspaces of $M (d,\mathbb{C})$. It is also known that the maximal number of MUUBs that can exist for $d^2$ dimensional space is $d^2-1$ though it remains unknown if such a number can be constructed \cite{Scott}. As a matter of fact the question of how many can even exist for the smaller $d$ dimensional subspace was unanswered for any prime $d$, let alone a proper recipe for construction. In this work, we provide definitive answers to these questions. We should like to point out that in ref.\cite{Jesni}, the construction of MUUBs for the $d$ dimensional subspace was done for only $d=2,3$ where $d=2$ was simple enough to be solved analytically based on the definition of MUUBs, as it dealt with only $2\times 2$ matrices. The case for $d=3$ on the other hand was based on some numerical search.

We begin in Sec. II with the definitions of MUB and MUUB and the relevant focus of the work; i.e. prime dimensional Hilbert spaces and  the $d$ dimensional subspace of $M (d,\mathbb{C})$. Given the relevant isomorphism between finite dimensional vector spaces with the same dimensionality, we show in Sec. III  that a unique linear transformation maps MUBs for $\mathcal{H}_d$ to bases for some subspace of $M (d,\mathbb{C})$ such that, the Hilbert-Schmidt norm for an element in one such bases with one from another is proportional to the inner product of states from different MUBs. In Sec. IV, we introduce a binary operation on $\mathcal{H}_d$ to form a monoid which is then shown to be isomorphic to a monoid defined based on the relevant subspace of $M (d,\mathbb{C})$. This is then used to provide the proof for the maximal number of MUUBs as being less than 1 compared to the MUBs for $\mathcal{H}_d$. We note that the proof is constructive in the sense that it also provides a recipe for the construction of such MUUBs. We then use this to construct MUBs consisting of maximally entangled states (MES) for a subspace of $\mathcal{H}_d\otimes \mathcal{H}_d$ before ending the work with a concluding section.

\section{MUB and MUUB}

\noindent Let us begin with definitions of MUBs  \cite{Band} and MUUBs  \cite{Jesni}.\\
\begin{definition}
Two distinct orthonormal bases for a $d$-dimensional Hilbert space, ${ B^{(0)}}=\{ \vert \varphi_0 \rangle,...,\vert \varphi_{d-1} \rangle \}$ and ${ B^{(1)}}=\{ \vert \phi_0 \rangle,...,\vert \phi_{d-1} \rangle \}$ are said to be mutually unbiased bases (MUB) provided that $\vert \langle \varphi_i \vert \phi_j \rangle\vert=1/\sqrt{d}$, for every $i, j=0,...,{d-1}$.
\end{definition}

\noindent {  Superscripts in parentheses represent a basis's index (a notation we shall use throughout)}. We shall focus only on MUBs for $\mathcal{H}_d$ with $d$ being a prime number greater than 2. Beginning with the computational basis, $\{|0\rangle, ..., |d-1\rangle\}$, the $s$-th state from any of the remaining $d$ MUBs can thus be written as 
\begin{eqnarray}\label{ket}
\vert { \omega}_{s}^{(r)} \rangle=\frac{1}{\sqrt{d}} \sum_{a=0}^{d-1} ({ \omega}^s)^{d-a} ({ \omega}^{-r})^{\alpha_{(a)}} \vert a \rangle
\end{eqnarray}
where ${ \omega}$ is $d$th root of unity, ${ \omega}=\exp ^{2 \pi i /d}$, $s=0,...,d-1$, $r=0,...,d-1$ (the index $r$ { and s} indicating the $r$-th basis { and element of basis respectively}) and $\alpha_{(a)} =a+..+(d-1)$.{ It is instructive to note that $\omega^i\omega^j=\omega^{i\oplus  j}$ with $\oplus$ as addition modulo ${d}$.}\\

\noindent
\begin{definition}
{ Two distinct orthogonal bases composing of unitary transformations, $A^{(0)}=\{A_0^{(0)},...,A_{n-1}^{(0)}\}$ and $A^{(1)}=\{A_{0}^{(0)},...,A_{n-1}^{(1)}\}$,  for some $n$ dimensional subspace of the vector space $M (d,\mathbb{C})$ are MUUB provided that,
\begin{eqnarray}
\vert \Tr (A_i^{(0)\dagger} A_j^{(1)}) \vert ^2=C,~\forall A_i^{(0)} \in A^{(0)}, A_j^{(1)} \in A^{(1)}
\end{eqnarray}
for $i, j=0,...,n-1$ and some constant $C \neq 0$}.
\end{definition}
In what follows, we shall limit ourselves to the $d$ dimensional subspace of $M (d,\mathbb{C})$ with $d$ being a prime number for which its standard basis can be written as
\begin{eqnarray}\label{X}
\{X^0, X^1,...,X^{d-1}\}
\end{eqnarray}
with $X^0=\mathbb{I}_d$, the identity operator and $X^i X^j=X^{i\oplus j}$. Note that $\vert\Tr(X^i X^{j\dagger})\vert =0$ for $i\neq j$. These operators may correspond to the generalised Pauli operators. For simplicity, we shall refer to such a subspace as $M_s$. 

We should mention that, strictly speaking, $\mathcal{H}_d$ (and $M_s$ for the matter) are inner product spaces, i.e. vector spaces endowed with an inner product each. Therefore they are mathematically distinct from their underlying vector spaces (i.e. without defining the inner product on the vector space) or even the underlying sets. However, in the following, we will use $\mathcal{H}_d$ and $M_s$ to refer to anyone of these structures and which is referred to would be clear from context.
\section{The inner product spaces of $\mathcal{H}_d$ and $M_s$.}
 
{ Let us define a map $G:\mathcal{H}_d\rightarrow M_s$} with the rule
\begin{eqnarray}
G(\sum_ia_i|i\rangle)=\sum_ia_i X^i~,~\forall a_i\in \mathbb{C}, i=0,...,d-1.
\end{eqnarray}
{ It can be easily shown that $G$ is an isomorphism mapping any bases for $\mathcal{H}_d$ to one for $M_s$\footnote{This can be understood as a composition of the isomorphism from $\mathcal{H}_d$ to $\mathbb{C}^d$ and from the latter to $M_s$}}. Ensuring that MUBs are mapped to bases for $M_s$ which are `mutually unbiased' to one another, we give the following lemma.
\begin{lemma}\label{L3}
Consider any two states $|\phi \rangle,|\varphi \rangle\in\mathcal{H}_d$ and $G ( \vert \phi \rangle ),G ( \vert \varphi \rangle )\in M_s$. The absolute value of the inner product on $\mathcal{H}_d$, $\vert \langle \phi\vert \varphi \rangle \vert$, is proportional to the absolute value of the inner product on $M_s$, $\vert Tr [ G ( \vert \phi \rangle )^\dagger  G(\vert \varphi\rangle)] \vert$, i.e. $\vert \langle \phi\vert \varphi \rangle \vert=(1/d)(\vert Tr [ G ( \vert \phi \rangle )^\dagger G(\vert \varphi\rangle)] \vert)$.
\end{lemma}
\noindent
\begin{proof}
Let 
$\vert \varphi \rangle=\sum_{i} a_i \vert i \rangle$ and
$\vert \phi \rangle=\sum_{j} b_j \vert j \rangle$. Hence,
we write 
\begin{eqnarray}
G (\vert \varphi \rangle)=G (\sum_{i} a_i \vert i \rangle)=\sum_{i} a_i X^i,\nonumber\\
G (\vert \phi \rangle)=G (\sum_{j} b_j \vert j \rangle)=\sum_{j} b_j X^j.
\end{eqnarray}
Then, 
\begin{align}
\vert \text{Tr} [ G ( \vert \phi \rangle )^\dagger G(\vert \varphi \rangle)] \vert
=&|\text{Tr}[\sum_{j} b_j^\ast X^{j\dagger}\sum_{i} a_i X^i]|\nonumber\\
=&\vert \text{Tr}[\sum_{i,j} a_i b_j^\ast X^i X^{j\dagger}]\vert~.
\end{align}

\noindent
As $\vert \text{Tr}(X^i X^{j\dagger}) \vert=d \delta_{ij} $, we have
\begin{equation}
\vert \text{Tr}[ G ( \vert \phi \rangle )^\dagger G(\vert \varphi \rangle)] \vert= d (\sum_{i,j} a_i b_j^\ast)\\\nonumber
=d\vert \langle \phi \vert \varphi \rangle \vert~.
\end{equation}
%
\qed 
\end{proof}
Admittedly, this lemma may be a long winded way to describe an \textit{isomorphism} between inner product spaces. However, we chose to be explicit mainly for the reason that the above lemma does not give an isometry (which is what is conventionally presented for isomorphism between inner product spaces). Further to that, the mapping between the elements of the underlying set would continue to be useful in the later sections. { An immediate corollary is that, if an operator $U\in M_s$ is unitary, then $|\text{Tr}[ U^\dagger U]|=d$. Thus, $G^{-1}(U)=|U\rangle$ is a normalised vector, i.e. $\vert \langle U \vert U \rangle \vert=1$. The reverse however, is not generally true (it is certainly possible to find operators such that the inner product with itself results in $d$, though not necessarily unitary).}
\begin{lemma}\label{L4}
{ Let $B^{(0)} = \{ \vert \varphi_{0} \rangle,..., \vert \varphi_{d-1} \rangle \}$ and $B^{(1)}=\{ \vert \phi_{0} \rangle,..., \vert \phi_{d-1} \rangle \}$ be two bases for $\mathcal{H}_d$. $B^{(0)}$ and $B^{(1)}$ are MUB to one another if and only if the sets $\{ G(\vert \varphi_{0} \rangle),...,G(\vert \varphi_{d-1} \rangle)\}$ and $\{ G(\vert \phi_{0} \rangle),...,G(\vert \phi_{d-1} \rangle)\}$ are bases for $M_s$ such that $\vert \text{Tr} [ G ( \vert \phi_i \rangle )^\dagger G(\vert \varphi_i \rangle)] \vert^2=d.$}
\end{lemma}
\begin{proof}
The proof follows quite directly from the previous lemma. Beginning with $B^{(0)}$ and $B^{(1)}$ being MUB to one another, $\vert \langle \phi_i\vert \varphi_j\rangle \vert^2=1/d$, the previous lemma gives
\begin{align}
\vert \langle \phi_i \vert \varphi_j \rangle \vert^2=&\left(\frac{1}{d} \vert\text{Tr} [ G ( \vert \phi_i \rangle )^\dagger G(\vert \varphi_j \rangle)] \vert\right)^2
\nonumber\\
\therefore d =& \vert\text{Tr} [ G ( \vert \phi_i \rangle )^\dagger G(\vert \varphi_j \rangle)] \vert^2.
\end{align}
{ The proof for the reverse can be shown quite easily given that $G^{-1}$ is also an isomorphism.}\qed
\end{proof}
While this tells us of the relationship between the bases of $M_s$ arrived through the map $G$, and thus there would be $d+1$ such bases, it provides no clue as to whether the bases consists of unitary operators or otherwise. In order to determine that, we shall thus first define a particular binary operation on $\mathcal{H}_d$ and then resort to an isomorphism between the monoids  $\mathcal{H}_d$ and $M_s$.
\section{The Monoids $\mathcal{H}_d$ and $M_s$.}

\noindent Let us define a binary operation, $\bullet$, on the set $\mathcal{H}_d$. 
\begin{definition}
\noindent
For all $\vert \varphi \rangle,\vert \phi \rangle \in \mathcal{H}_d $ given as 
\begin{eqnarray}
\vert \varphi \rangle=\sum_{i} a_i \vert i \rangle
~,~\vert \phi \rangle=\sum_{j} b_j \vert j \rangle
\end{eqnarray}
with $a_i,b_j\in\mathbb{C}$ and  $i,j=0,...,d-1$, we let $\bullet:\mathcal{H}_d\times \mathcal{H}_d\rightarrow \mathcal{H}_d$ be defined as
\begin{equation}
\sum_{i} a_i \vert i \rangle \bullet \sum_{j} b_j \vert j \rangle = \sum_{m} \alpha_m \vert  m \rangle
\end{equation}
{ where $\alpha_m=\sum_{i,j} a_i b_j$ and $m=i\oplus j$}.
\end{definition}

We show in Appendix A,  that $\mathcal{H}_d$ with the operation of $\bullet$ defines a monoid. In principle, a monoid can be written as a triple, i.e. $(\mathcal{H}_d,\bullet,e)$ where $e$ is the identity of the monoid. The monoid $(M_s, \cdot,\mathbb{I}_d)$ has the subset $M_s$ of the set $M(d,\mathbb{C})$ as its underlying set with the usual matrix multiplication, $\cdot$, as the binary operation and $\mathbb{I}_d$ the identity. However in what follows, when its clear from context that we are referring to monoids, we shall, for simplicity omit the triple notation and write $\mathcal{H}_d$  and $M_s$ instead. Further to that, we shall also omit the symbol `$\cdot$' when its clear from the context.

\begin{lemma}\label{L6}
The monoid $\mathcal{H}_d$ is isomorphic to $M_s$; i.e., there exists a bijection $G:\mathcal{H}_d\rightarrow M_s$ such that
 $\forall\vert \varphi \rangle ,\vert \varphi \rangle \in \mathcal{H}_d$, 
 \begin{eqnarray}
 G (\vert \varphi \rangle \bullet \vert \phi \rangle)=G(\vert \varphi \rangle) { \cdot} G(\vert \phi \rangle)
 \end{eqnarray}
 \end{lemma}
\begin{proof}
Consider
\begin{equation}
\vert \varphi \rangle=\sum_{i} a_i \vert i \rangle~~
\end{equation}
\noindent
and
\begin{eqnarray}
\vert \phi \rangle=\sum_{j} b_j \vert j \rangle~.
\end{eqnarray}
Let $G (\vert \varphi \rangle)=G (\sum_{i} a_i \vert i \rangle)=\sum_{i} a_i X^i$ with $G (\vert \phi \rangle)=G (\sum_{j} b_j \vert j \rangle)=\sum_{i} b_j X^{j}$ and $G(\sum_{i,j} a_i b_j \vert i \oplus j \rangle)=\sum_{i,j} a_i b_j X^{i \oplus j}$, { then}
\begin{align}
G(\vert \varphi \rangle \bullet \vert \phi \rangle)=&G(\sum_{i} a_i \vert i \rangle \bullet \sum_{j} b_j \vert j \rangle)\nonumber\\=&G(\sum_{i,j} a_i b_j \vert i \oplus j \rangle)\nonumber\\
=&\sum_{i,j} a_i b_j X^{i \oplus j} \nonumber\\ =&\sum_{i} a_i X^i.\sum_{i} b_j X^{j} \nonumber\\ =&G (\vert \varphi \rangle). G (\vert \phi \rangle)~.
\end{align}\qed
\end{proof}
{ \noindent We now define a useful concept, ``conjugate transpose" of a \textit{ket} $\vert \varphi \rangle=\sum_{i} a_i \vert i \rangle$. 
\begin{definition}\label{defC}
The conjugate transpose of a state $\vert \varphi \rangle=\sum_{i} a_i \vert i \rangle\in \mathcal{H}_d$ is given as $\vert \varphi \rangle^\dagger=\sum_{i} a^*_i \vert d-i \rangle$.
\end{definition}
This is different from the conventional definition in terms of the dual of a ket, i.e. $|\phi \rangle^\dag = \langle\phi |$. In other words, the standard conjugate transpose of a column vector gives the row vector with conjugated components. Our definition} of such a state is in fact motivated by the conjugate transpose of the operator, say, $M=\sum a_iX^i$. In other words, if $G^{-1}(M)=\sum_{i} a_i \vert i \rangle=\vert \varphi \rangle$, then $G^{-1}(M^\dagger)=\sum_{i} a^*_i \vert d-i \rangle=\vert \varphi \rangle^\dagger$. \\
\begin{lemma}\label{L7}{
$\vert \varphi \rangle \bullet \vert \varphi \rangle^\dagger=\vert 0 \rangle$ if and only if $G(\vert \varphi \rangle ) \cdot G (\vert \varphi \rangle^\dagger)=\mathbb{I}_d$ . }
\end{lemma}
\begin{proof}
\noindent
This proof follows straightforwardly from { Lemma \ref{L6} on the isomorphism between $\mathcal{H}_d$ and $M_s$. That is,}
\begin{align}
G (\vert \varphi \rangle \bullet \vert \varphi \rangle^\dagger)=&G(\vert 0 \rangle)~,
\nonumber\\
{ \therefore} ~G(\vert \varphi \rangle) { \cdot} G(\vert \varphi \rangle^\dagger)=&\mathbb{I}_d~.
\end{align}
The proof for the reverse is also straightforward and follows in a similar line to the above given the bijective nature of isomorphisms.\qed
\end{proof}
{ This lemma gives us a way to determine if $G$ maps a state in $\mathcal{H}_d$ to an unitary operator in $M_s$. We shall need one more lemma before we proceed to a final theorem. 
\begin{lemma}\label{L8}
The maximal number of MUUBs for a $d$-dimensional subspace of $M(d,\mathbb{C})$ is strictly less than $d+1$.
\end{lemma}
\begin {proof}
Given the Choi isomorphism  for the equivalence between unitary operators and maximally entangled states, the search for MUUBs is equivalent to the search for MUBs consisting of MES.
For a unitary $U_i$, its equivalent MES, $|\mathcal{U}_i\rangle$, can be written as \cite{sac,d1,d2}, 
\begin{eqnarray}\label{udd}
|\mathcal{U}_i\rangle=(\sum_{m}\sum_n \langle n|U_i|m\rangle |m\rangle|n\rangle)/\sqrt{d}~.
\end{eqnarray}
with $|m\rangle$,$|n\rangle$ as some basis vectors for $\mathcal{H}_d$. 
 As we are now considering only the $d$ dimensional subspace of $M(d,\mathbb{C})$, or equivalently, the $d$-dimensional subspace of $\mathcal{H}_d\otimes\mathcal{H}_d$, the inner product between MES coming from differing MUBs (assuming their existence) in this subspace becomes $1/\sqrt{d}$. This corresponds to the inner product of operators from two MUUBs being equal to $d$ (Lemma \ref{L4}) and was noted explicitly for $d=2$ in ref.\cite{Jesni}.

Let us write a basis for this $d$-dimensional subspace of $\mathcal{H}_d\otimes\mathcal{H}_d$ as $\{|\mathcal{U}_i\rangle\}_{i=0}^{d-1}$ for which the subspace spanned is isomorphic to $\mathbb{C}^d$. Writing the basis for the latter as  $\{|i_U\rangle\}_{i=0}^{d-1}$ with $|\mathcal{U}_i\rangle\equiv|i_U\rangle$, the problem essentially reduces to the search for MUBs of a $d$-dimensional space. 
We follow closely the works of \cite{Scott,Hall} where \cite{Scott} gives the maximal number of MUUBs for $M(d,\mathbb{C})$ and \cite{Hall} which gives a simple description how the standard $d+1$ number of MUBs is maximal for a $d$-dimensional Hilbert space.

Considering the set $\{|i_U\rangle\langle j_U|~|~i,j=0,...,d-1\}$ as an orthonormal basis for the $d^2$ space of $\text{End}(\mathbb{C}^d)$, any state (corresponding to either MES or non-MES included\footnote{We should note that the word `corresponding' is used to highlight the simple fact that one cannot determine the nature of entanglement simply based on states in $ \text{End}(\mathbb{C}^d)$}) can thus be written as $\rho_d\in \text{End}(\mathbb{C}^d)$ such that $\text{Tr}(\rho_d)=1$. The dimensionality of this space, $\text{End}(\mathbb{C}^d)$, is obviously $d^2$. However, if we are interested only in the smallest subspace containing states corresponding to MES, call it $H_{mes}$, its dimensionality, $\text{dim}~H_{mes}$ is strictly lesser than $d^2$ (we show this in the Appendix B). Thus if we consider the subspace of traceless Hermitian operators thereof,  $H_{t}$,
\begin{eqnarray}  
H_t=\left\{\rho-\mathbb{I}_d/d~|~\rho~\text{corresponds to MES}\right\}
\end{eqnarray}
the dimensionality, $\text{dim}~H_t< d^2-1$ . For any MUB comprising of $d$ (MES) states, one can construct a set of $d-1$ traceless Hermitian operators \cite{Hall}. It can be shown that for any two states, say $\rho_A,\rho_B\in H_{mes}$ coming respectively from two distinct MUBs, 
\begin{eqnarray}\label{ort}
\text{Tr}[(\rho_A-\mathbb{I}_d/d)(\rho_B-\mathbb{I}_d/d)]=0.
\end{eqnarray}
This implies the subspaces are orthogonal to one another. Hence the maximal number of MUBs consisting of MES is $(\text{dim}~H_t)/(d-1)<d+1$. Going back to $d$-dimensional subspace of $M(d,\mathbb{C})$, the maximal number of MUUB is thus lesser than $d+1$.\\
\qed
\end{proof}

}

\begin{theorem}
The maximal number of MUUBs for the $d$ dimensional subspace $M_s$ is $d$.
\end{theorem}
\begin{proof}

\noindent
In the first part of the proof, we will show that while the isomorphism between the vector space $H_d$ and $M_s$ maps between the bases of the vector spaces, at least one of the $d+1$ MUBs would be mapped to a basis of $M_s$ which would contain a operator which is non-unitary.

Consider $\vert { \omega}_{s}^{(r)}\rangle$ given by eq.(\ref{ket}). Then, consider its conjugate transpose according to definition \ref{defC}, $\vert { \omega}_{s}^{(r)} \rangle^\dagger$, as 
{\begin{eqnarray}
\vert { \omega}_{s}^{(r)} \rangle^\dagger=\frac{1}{\sqrt{d}} \sum_{a=0}^{d-1} [({ \omega}^s)^{d-a} ({ \omega}^{-r})^{\alpha_{(a)}}]^\ast \vert d-a \rangle~.
\end{eqnarray}}
Thus to determine if $\vert { \omega}_{s}^{(r)}\rangle$ would be mapped to a unitary operator, we consider the following,
\begin{align}
\vert { \omega}_{s}^{(r)} \rangle \bullet \vert { \omega}_{s}^{(r)}\rangle^\dagger
=& \frac{1}{d} \sum_{m=0}^{d-1}\left [ \sum_{a=0}^{d-1} ({ \omega}^s)^{d-a} (({ \omega}^s)^{b_a^{(m)}})^\ast \cdot ({ \omega}^{-r})^{\alpha_{(a)}} (({ \omega}^{-r})^{\alpha_{(d-b_a^{(m)})}})^\ast\right] \vert m \rangle\nonumber\\
=&\frac{1}{d} \sum_{m=0}^{d-1} \left[ \sum_{a=0}^{d-1} ({ \omega}^{s(d-a-b_a^{(m)})}) \cdot({ \omega}^{\frac{1}{2}r (d-b_a^{(m)}-(d-b_a)^2+(a-1)a)}) \right]  \vert m \rangle
\end{align}
%
%
with the integer $b_a^{(m)}\in[0,d-1]$ such that $a\oplus b_a^{(m)}=m$. We made use of the fact, $\alpha_{(d-b_a^{(m)})}-\alpha_{(a)}= [(d- b_a^{(m)})-(d- b_a^{(m)})^2+(a-1)a]/2$. { As $m=a\oplus b_a^{(m)}$, thus $b_a^{(m)}=m+qd-a$ for some integer $q\ge 0$. As $a$ and $b_a^{(m)}$ are both lesser than $d-1$, thus $q$ cannot be greater than 1.}
With $\omega^d=\omega^{b_a^{(m)}d}=1$ and $\omega^{d/2}=\omega^{d^2/2}=-1$,  we can further simplify the above into
{\begin{eqnarray}\label{ww}
\vert { \omega}_{s}^{(r)} \rangle \bullet \vert { \omega}_{s}^{(r)}\rangle^\dagger
= \frac{1}{d} \sum_{m=0}^{d-1}\left[ \sum_{a=0}^{d-1}{ \omega}^{arm}\cdot { \omega}^{-\frac{r}{2}(m^2+m)-ms } \right] \vert m \rangle.
\end{eqnarray}
Consider the coefficient for the ket $|m\rangle$ is given by 
\begin{eqnarray}\label{cw}
\left[\dfrac{ \omega^{-\frac{r}{2}(m^2+m)-ms }}{d}\right ]\sum_{a=0}^{d-1}{ \omega}^{arm}~,
\end{eqnarray}
in} the case of $r=s=0$, we have, the coefficients for every ket (irrespective of $m$) being 1. This implies 
\begin{eqnarray}
\vert { \omega}_{0}^{(0)} \rangle \bullet \vert { \omega}_{0}^{(0)} \rangle^\dagger=\sum_{m=0}^{d-1}\vert m \rangle\neq |0\rangle~.
\end{eqnarray}
Thus, the element $\vert { \omega}_{0}^{(0)} \rangle$ cannot be mapped to a unitary transformation under $G$, hence at least 1 MUB of $\mathcal{H}_d$ cannot be mapped to a unitary basis for $M_s$.\\
\newline
\noindent The second part of the proof addresses the case for the MUBs with $r\neq 0$; where we will show that only in the case for $m=0$, the coefficient becomes $1$ (irrespective of $r$) and zero otherwise. The former is straightforward and can be seen by setting $m=0$.

{ We can rewrite the index of $\omega$ in the summation of eq.(\ref{cw}), $arm$ simply as $ap$ with $p=rm$. It is obvious that $p$ is an integer}. Writing $l=ap\mod{d}$, for every integer value of $a\in[0,d-1]$, $l$ will take on a unique integer value in $[0,d-1]$ (this is shown in Appendix C). Thus, for $m \neq 0$,
\begin{equation}
{\sum_{a=0}^{d-1} { \omega}^{(ap)}=\sum_{l=0}^{d-1} { \omega}^l =0~.}
\end{equation}
Thus, we have $\vert { \omega}_{s}^{(r)} \rangle \bullet \vert { \omega}_{s}^{(r)} \rangle^\dagger
=|0\rangle$ for $r\neq 0$. Together with Lemma 5, we conclude that only $d$ number of MUBs of $\mathcal{H}_d $ are mapped to MUUBs of $M_s$. 

\qed
\end{proof}

\subsection{Recipe for constructing MUUBs for $M_s$}
\noindent It is now quite obvious that we can construct the maximal number of MUUBs for { the subspace, $M_s$, as described in  the preceeding sections. Writing a basis for the subspace as $X^{(0)}=\{\mathbb{I}_d,X,X^2,...,X^{d-1}\}$, an $s$-th operator of an $r$-th basis, $X_s^{(r)}$, which is MUUB with respect to $X^{(0)}$ is given by}
\begin{eqnarray}
X_s^{(r)}=\dfrac{1}{\sqrt{d}} \sum_{i=0}^{d-1} ({ \omega}^s)^{d-i} ({ \omega}^{-r})^{\alpha_{(i)}} X^i 
\end{eqnarray}
with $r=1,...d-1, s=0,...,d-1,$ and $\alpha_{(i)}=i+...+d-1$. { The bases constructed thus are also pairwise mutually unbiased.}
{ We provide an explicit example for the 3-dimensional subspace of $M(3,\mathbb{C})$. Beginning with say a basis, $\mathcal{X}^{(0)}=\{ \mathbb{I}_3, \mathbb{X}_3, \mathbb{X}_3^2\}$ where $\mathbb{X}_3$ is the generalised Pauli matrix for 3 dimensions (the subscript `3' reflects the 3 dimensional Hilbert space the operator acts on), the other two sets of MUUBs, $\{ \mathcal{X}_0^{(1)}, \mathcal{X}_1^{(1)}, \mathcal{X}_2^{(1)} \}$ and $\{ \mathcal{X}_0^{(2)}, \mathcal{X}_1^{(2)}, \mathcal{X}_2^{(2)} \}$ are described by the operators,}
{ \begin{align}
\mathcal{X}_m^{(1)} &=\frac{1}{\sqrt{3}} [ \mathbb{I}_3 + \omega^{2m} \mathbb{X}_3 +\omega^{m+1} \mathbb{X}_3^2 ]~,
\nonumber\\
\mathcal{X}_n^{(2)} &=\frac{1}{\sqrt{3}} [ \mathbb{I}_3 + \omega^{2n} \mathbb{X} _3+\omega^{n+2} \mathbb{X}_3^2 ]~.
\end{align} \newline
\noindent The above gives a total of 3 MUUBs; we note that this corresponds to the result of the numerical search in ref.\cite{Jesni}.
}
\section{Connection to maximally entangled states}
{ \noindent  The Choi isomorphism between unitary operators on a $d$ dimensional Hilbert space, $\mathcal{H}_d$, and vectors in $\mathcal{H}_d\otimes \mathcal{H}_d$ was used in ref.\cite{Jesni} as a direct way of constructing MUBs consisting of maximally entangled states (then, explicitly for $d=2$). 
\noindent
In constructing MUBs for MES of $\mathcal{H}_d\otimes \mathcal{H}_d$, let us begin with the unitary operators as generalised Pauli matrices, $\mathbb{X}_d$ and $\mathbb{Z}_d$ \cite{Hall} corresponding to MES (as in equation (\ref{udd})) giving the generalised Bell states  as in Ref. \cite{Bennet,Sych},
\begin{eqnarray}
\frac{1}{\sqrt{d}} \vert \mathbb{X}_{d}^{b}\mathbb{Z}_{d}^{a}\rangle=\frac{1}{\sqrt{d}} \sum_{n=0}^{d-1} \omega^{a n} \vert n \rangle \vert n \oplus b \rangle
\end{eqnarray} 
for $a,b=0,...,d-1$. Note the subscript $d$  represents the dimensionality of the Hilbert space the operator acts on. We then choose, only $d$ elements with some fixed values of $a,b$ corresponding to the MES of the form $\vert (\mathbb{X}_{d}^{b}\mathbb{Z}_{d}^{a})^i\rangle/\sqrt{d}$ for $i=0,...,d-1$. This would be a basis for a $d$ dimensional subspace. It is worth noting that this choice of operator, $\mathbb{X}_{d}^{b}\mathbb{Z}_{d}^{a}$, do fulfil our `requirement' of eq. (3). We show this in Appendix D. The $s$-th state for $r$-th basis can then be written as
\begin{eqnarray}
\vert \mathbb{X}_{d}^{b}\mathbb{Z}_{d}^{a}\rangle^{(r)}_s=\dfrac{1}{d}\sum_{i=0}^{d-1} ({ \omega}^s)^{d-i} ({ \omega}^{-r})^{\alpha_{(i)}} \vert (\mathbb{X}_{d}^{b}\mathbb{Z}_{d}^{a})^i\rangle\nonumber\\
\end{eqnarray} 
with $r=1,...,d-1$ and $s=0,...,d-1$.}
\section{Conclusion}

{ 
\noindent
Drawing on the \textit{similarities} between the Hilbert space, $\mathcal{H}_d$, and some subspace, $M_s$, of $M (d,\mathbb{C})$, we have shown that the number of MUUBs that can be constructed for $M_s$, is $d$, i.e. one less than the number of MUBs for $\mathcal{H}_d$ being $d+1$. More specifically, we had considered the isomorphism between the vector spaces, then defined a monoid based on $\mathcal{H}_d$ and subscribe to an isomorphism between it and that of $M_s$ to result in a constructive theorem for the matter.

We also considered the isomorphism between unitary operators and entangled states to provide for a construction of MUBs for MES constrained to some $d$ dimensional subspace of $\mathcal{H}_d\otimes\mathcal{H}_d$. The selection of MES for the construction was based on the unitaries used in accordance with the set in eq.(\ref{X}). In fact, one can easily find examples to demonstrate that the selection of MES for such a construction cannot be arbitrary, i.e. a subspace which contains MUBs must see its bases composed of selected MES with some clear rule.

The work marks a clear milestone in the direction of constructing the maximal number of MUUBs for the more general scenario, i.e. for the entire space of $M(d,\mathbb{C})$ of operators on $\mathcal{H}_d$, or equivalently, MUBs for MES in $\mathcal{H}_d\otimes\mathcal{H}_d$.

}

 \begin{appendices}
\section{The monoid $(\mathcal{H}_d,\bullet,|0\rangle)$}
We show that $(\mathcal{H}_d,\bullet,|0\rangle)$ is a monoid. We begin first by showing that $\bullet$ is closed on $\mathcal{H}_d$. \\
\newline
\noindent $ \vert \varphi \rangle \bullet \vert \phi \rangle=\sum_{i} a_i \vert i \rangle \bullet \sum_{j} b_j \vert j \rangle=\sum_{m} \alpha_m \vert m \rangle $. Note that $\alpha_m \in \mathbb{C}~,  \forall ~m$ , therefore $\sum_{m} \alpha_m \vert m \rangle \in \mathcal{H}_d$, hence the set is closed under the operation $\bullet$.\\
\newline
We next show that the operation is associative on $\mathcal{H}_d$.
%
%
\begin{eqnarray}
(\vert \varphi \rangle \bullet \vert \phi \rangle) \bullet \vert \Phi \rangle=\sum_{m} \alpha_m \vert m \rangle \bullet \sum_{k} c_k \vert k \rangle
\end{eqnarray}
where $\alpha_m=\sum_{i,j} a_i b_j ~and~ i \oplus j =m$. Then,

\begin{eqnarray}
\sum_{m} \alpha_m \vert m \rangle \bullet \sum_{k} c_k \vert k \rangle=\sum_{n} \beta_n \vert n \rangle
\end{eqnarray}
where $\beta_n=\sum_{m,k} \alpha_m c_k	~ and~ m \oplus k =n$ or $\beta_n=\sum_{i,j,k}  a_i b_j c_k$ with $i \oplus j \oplus k =n$.

\begin{align}
\sum_{n} \beta_n \vert n \rangle=&\sum_{i, j, k} a_i b_j c_k \vert i \oplus j \oplus k \rangle\nonumber\\
=&\sum_{i,j,k} ( a_i)( b_j c_k) \vert i \oplus j \oplus k \rangle \nonumber\\=&(\sum_{i} a_i \vert i \rangle) \bullet (\sum_{j,k} b_j c_k \vert j \oplus k \rangle)
\nonumber\\
=&\vert \varphi \rangle \bullet (\vert \phi \rangle \bullet \vert \Phi \rangle)
\end{align}
where $\vert \phi \rangle \bullet \vert \Phi \rangle=\sum_{j,k} b_j c_k \vert j \oplus k \rangle$.
Hence, $\bullet$ is associative.

To complete the requirement of being a monoid, we show the existence of an identity in $\mathcal{H}_d$ with respect to the operation $\bullet$.
Let $\vert e \rangle=\sum_{i} a_i \vert i \rangle$ and $\sum_{i,j} a_i b_j  \vert i \oplus j \rangle$ and $ i \oplus j=m$. Then, $\vert e \rangle$ is an identity if
\[
  a_i =
  \begin{cases}
                                   1 & \text{if $i=0$} \\
                                   0 & \text{if $i \neq 0$} ~.
  
  \end{cases}
\]
This effectively means $|e\rangle=|0\rangle$. We write it in the above form to ensure that the operation with any state in $\mathcal{H}_d$ follows from the definition for $\bullet$.
\noindent
Thus
$\vert e \rangle \bullet \sum b_j \vert j \rangle=\sum_{i} a_i \vert i \rangle \bullet \sum b_j \vert j \rangle = \vert 0 \rangle \bullet \sum b_j \vert j \rangle =\sum b_j  \vert 0 \oplus j \rangle=\sum b_m \vert m \rangle$. On the other hand, 
$\sum b_j \vert j \rangle \bullet \vert e \rangle=\sum b_j \vert j \rangle \bullet \sum_{i} a_i \vert i \rangle =\sum b_j \vert j \rangle \bullet \vert 0 \rangle=\sum b_j  \vert j \oplus 0 \rangle= \sum b_m \vert m \rangle$ . Hence $\vert e \rangle=\vert 0 \rangle$ is an identity element of $\mathcal{H}_d$.

In the following, we show that the operation is also distributive on $\mathcal{H}_d$ { with respect to vector sum. Consider}
\begin{align}
\vert \varphi \rangle \bullet (\vert \phi \rangle + \vert \Phi \rangle)=&\sum_{i} a_i \vert i \rangle \bullet (\sum_{j} b_j \vert j \rangle \nonumber\\+& \sum_{j} c_j \vert j \rangle)=\sum_{s} \omega_s \vert s \rangle
\end{align}
where $\omega_s=\sum_{i,j} a_i (b_j + c_j)$ and $i \oplus j = s$.

\begin{align}
\sum_{s} \omega_s \vert s \rangle=&\sum_{i,j} a_i (b_j + c_j) \vert i \oplus j \rangle\nonumber \\=&\sum_{i,j} a_i b_j \vert i \oplus j \rangle + \sum_{i,j} a_i c_j  \vert i \oplus j \rangle\nonumber\\
=&(\sum_{i} a_i \vert i \rangle \bullet \sum_{j} b_j \vert b_j \rangle )\nonumber\\
+&( \sum_{i} a_i \vert i \rangle \bullet \sum_{j} c_j \vert j \rangle)\nonumber\\
=&(\vert \varphi \rangle \bullet \vert \phi \rangle) +( \vert \varphi \rangle \bullet \vert \Phi \rangle)~.
\end{align}
Thus, $\bullet$ is distributive.
\section{}
As noted from the text, the subspace spanned by $\{|\mathcal{U}_0\rangle,...,|\mathcal{U}_{d-1}\rangle \}$ is isomorphic to $\mathbb{C}^d$ and a basis for the latter can be written as $\{|i_U\rangle\}_{i=0}^{d-1}$. Hence, $\{|i_U\rangle\langle j_U|~|~i,j=0,...,d-1\}$ is an orthonormal basis for the $d^2$ space of $\text{End}(\mathbb{C}^d)$. In order to ascertain if any of these states, $\rho\in\text{End}(\mathbb{C}^d)$ is (actually) a MES, we need to reconsider the subsystems, of the state. In this case, the subsystems are two $d$-dimensional states and we refer to the subsystems later as systems 1 and 2. 
Let us thus consider the entangled state from the Choi isomorphism from the unitary $U_i$ given as $|\mathcal{U}_i\rangle$, 
\begin{eqnarray}
|\mathcal{U}_i\rangle=(\sum_{m}\sum_n \langle n|U_i|m\rangle |m\rangle|n\rangle)/\sqrt{d}~.
\end{eqnarray}
 The set $\{|\mathcal{U}_i\rangle\langle \mathcal{U}_j|~|~i,j=0,...,d-1\}$ is an orthonormal basis for a $d^2$ subspace of $\text{End}(\mathbb{C}^d\otimes\mathbb{C}^d)$. Let us call this subspace $M_{d^2}$.\\
 \newline
We show that a function $F:\text{End}\left(\mathbb{C}^d\right)\rightarrow M_{d^2}$ with the rule
\begin{eqnarray}
F\left(\sum_{i,j}a_{i,j}|i_U\rangle\langle j_U|\right)=\sum_{i,j}a_{i,j}|\mathcal{U}_i\rangle\langle \mathcal{U}_j|
\end{eqnarray}
is an isomorphism.
Note that $F$ is injective,
\begin{align}
F\left(\sum_{i,j}a_{i,j}|i_U\rangle\langle j_U|\right)=F\left(\sum_{i,j}b_{i,j}|i_U\rangle\langle j_U|\right)\nonumber\\
\Rightarrow \sum_{i,j}a_{i,j}|\mathcal{U}_i\rangle\langle \mathcal{U}_j|=\sum_{i,j}b_{i,j}|\mathcal{U}_i\rangle\langle \mathcal{U}_j|\nonumber\\
\therefore \sum_{i,j}a_{i,j}|\mathcal{U}_i\rangle\langle \mathcal{U}_j|-\sum_{i,j}b_{i,j}|\mathcal{U}_i\rangle\langle \mathcal{U}_j|=0\nonumber\\
 \therefore\forall i,j,~a_{i,j}=b_{i,j}\Rightarrow \sum_{i,j}a_{i,j}|i_U\rangle\langle j_U|=\sum_{i,j}b_{i,j}|i_U\rangle\langle j_U|~.
\end{align}
The surjectiveness of $F$ is obvious (any element in $M_{d^2}$ can be written as $\sum_{i,j}k_{i,j}|\mathcal{U}_i\rangle\langle \mathcal{U}_j|=F\left(\sum_{i,j}k_{i,j}|i_U\rangle\langle j_U|\right)$. To show the linearity of $F$;
\begin{align}
F\left(\alpha\sum_{i,j}a_{i,j}|i_U\rangle\langle j_U|+\sum_{i,j}b_{i,j}|i_U\rangle\langle j_U|\right)\nonumber\\
=F\left(\sum_{i,j}(\alpha a_{i,j}+b_{i,j})|i_U\rangle\langle j_U|\right)=\sum_{i,j}(\alpha a_{i,j}+b_{i,j})|\mathcal{U}_i\rangle\langle \mathcal{U}_j|\nonumber\\
=\sum_{i,j}\alpha a_{i,j}|\mathcal{U}_i\rangle\langle \mathcal{U}_j|+b_{i,j}|\mathcal{U}_i\rangle\langle \mathcal{U}_j|\nonumber\\
=\alpha F\left(\sum_{i,j}a_{i,j}|i_U\rangle\langle j_U|\right)+F\left(\sum_{i,j}b_{i,j}|i_U\rangle\langle j_U|\right)~.
\end{align}
Let $\rho\in\text{End}(\mathbb{C}^d)$ and $\rho_M=F(\rho)\in\text{End}(\mathbb{C}^d\otimes\mathbb{C}^d)$. It is easy to see that to ascertain whether $\rho$ correspond to some MES, we consider if $\text{Tr}_1(\rho_M)=\text{Tr}_2(\rho_M)=\mathbb{I}_d/d$ where the subscripts $1$ and $2$ refer to the subsystems the partial trace is taken over. Let us refer to the space $H_{M}$ as the smallest subspace of Hermitian operators $M_{d^2}$ which contains all MES in $M_{d^2}$. Following \cite{Scott}, we can say that all MES in $M_{d^2}$ is a convex set of the subspace
\begin{eqnarray}
H_M=\{J\in\text{End}(\mathbb{C}^d\otimes\mathbb{C}^d)~|~\text{Tr}_1(J)=\text{Tr}_2(J)=\text{Tr}(J)\mathbb{I}_d/d\}~.
\end{eqnarray}
Note that every element in $M_{d^2}$ can be written as a $\sum r_{ik}\lambda_i\otimes\lambda_k$ where $\{\lambda_i\}_{i=0}^{d^2-1}$ is an orthonormal Hermitian operator basis for $\text{End}(\mathbb{C}^d)$ with $\lambda_0=\mathbb{I}/\sqrt{d}$ and $r_{ik}\in\mathbb{R}$. The subspace $H_M$ however requires $r_{k0}=r_{0k}=0$ which implies that every element in $H_M$ is also an element in $M_{d^2}$, but not vice versa. Thus $H_M$ is a proper subset of $M_{d^2}$; and $\text{dim}~H_M<d^2$.

Finally, given the isomorphism $F$, we can identify the smallest subspace of all operators in $\text{End}(\mathbb{C}^d)$ containing operators corresponding to MES, call it $H_{mes}$, with dimensionality equal to $\text{dim}~H_M$.

\section{}
{{ Let $d$ be some prime number, $a_1,a_2=0,1,...,d-1$ and $p$ be integers. Let $l=a_1p\mod{d}=a_2p\mod{d}$. There exists integers $n_1,n_2$ such that,}
\begin{align}
a_1 p-l=&n_1 d\nonumber\\
a_2 p-l=&n_2 d~.
\end{align}
With $p\neq 0$, we have
\begin{eqnarray}\label{a1a2}
(a_1-a_2)=\left(\dfrac{n_1-n_2}{p}\right )d~.
\end{eqnarray}
\noindent We shall show that the eq.(\ref{a1a2}) is only true when $a_1=a_2$.
{ Let us first consider the case $p > |n_1-n_2| $, implying that $(n_1-n_2)/p$ is a fraction. As $d$ is prime greater than $2$, therefore $a_1- a_2$ is a fraction. However, as $a_1, a_2 \in [0, d-1]$, therefore $a_1-a_2$ cannot be a fraction. Thus, $p$ cannot be greater than $|n_1-n_2|$. 

Now, If $p \leq |n_1-n_2|,$ then this implies that $|a_1-a_2| \geq d$. However, note that as $|a_1-a_2|< d-1$, hence $p$ is neither less than nor equal to $|n_1-n_2|$. Thus eq.(\ref{a1a2}) is true if $a_1-a_2=n_1-n_2=0$. In turn, if $a_1\neq a_2$, therefore we have $a_1p\mod{d}\neq a_2p\mod{d}$.
Further to that, $l=ap\mod{d}$ is defined for all values of $a=0,...,d-1$, thus $l=0,...,d-1$.}}
\section{}
 In showing that $(\mathbb{X}_d^b\mathbb{Z}_d^a)^i(\mathbb{X}_d^b\mathbb{Z}_d^a)^j=(\mathbb{X}_d^bZ_d^a)^{i\oplus j}$, we first need to prove the following identity,
\begin{eqnarray}
(\mathbb{X}_d^b\mathbb{Z}_d^a)^n=\omega^{ab(n^2-n)/2}\mathbb{X}_d^{bn}\mathbb{Z}_d^{an}~.
\end{eqnarray}

\noindent The proof is by induction. Its obvious that for $n=1$
\begin{eqnarray}
(\mathbb{X}_d^b\mathbb{Z}_d^a)^1=\omega^{ab(1^2-1)/2}\mathbb{X}_d^{b}\mathbb{Z}_d^{a}=\mathbb{X}_d^b\mathbb{Z}_d^a~.
\end{eqnarray}
Now assume this is true for $n=k$
\begin{eqnarray}
(\mathbb{X}_d^b\mathbb{Z}_d^a)^k=\omega^{ab(k^2-k)/2}\mathbb{X}_d^{bk}\mathbb{Z}_d^{ak}~,
\end{eqnarray}
we need to determine for $n=k+1$,
\begin{eqnarray}
(\mathbb{X}_d^b\mathbb{Z}_d^a)^{k+1}=\omega^{ab(k^2-k)/2}\mathbb{X}_d^{bk}\mathbb{Z}_d^{ak}\mathbb{X}^b\mathbb{Z}_d^a~.
\end{eqnarray}
As we swap the positions of $\mathbb{X}_d^b$ and $\mathbb{Z}_d^ak$, we will make a total number of $abk$ swaps and thus introduce the term $\omega^{abk}$. Hence
\begin{align}
(\mathbb{X}_d^b\mathbb{Z}_d^a)^{k+1}&=\omega^{ab(k^2-k)/2}\omega^{abk}\mathbb{X}_d^{b(k+1)}\mathbb{Z}_d^{a(k+1)}\nonumber\\
&=\omega^{ab(k^2-k+2k)/2}\mathbb{X}_d^{b(k+1)}\mathbb{Z}_d^{a(k+1)}\nonumber\\
&=\omega^{ab[(k+1)^2-(k+1)]/2}\mathbb{X}_d^{b(k+1)}\mathbb{Z}_d^{a(k+1)}~.
\end{align}
Now, if $n=d>2$ (the case for $d=2$ is trivial), we have\\
\newline
\begin{eqnarray}
(\mathbb{X}_d^b\mathbb{Z}_d^a)^d=\omega^{ab(d^2-d)/2}\mathbb{X}_d^{bd}\mathbb{Z}_d^{ad}~.
\end{eqnarray}
As $ab(d^2-d)/2=abd(d-1)/2$, and $d-1$ is an even number (because $d$ is prime), hence $abd(d-1)/2$ is some positive multiple, $M$, of $d$ or $abd(d-1)/2=Md$). Note that $\omega^{Md}=(\omega^d)^{M}=1$. Thus
\begin{eqnarray}
(\mathbb{X}_d^b\mathbb{Z}_d^a)^d=\mathbb{X}_d^{bd}\mathbb{Z}_d^{ad}=(\mathbb{X}_d^{d})^b(\mathbb{Z}_d^{d})^a=\mathbb{I}
\end{eqnarray}
as $\mathbb{X}_d^b=\mathbb{Z}_d^a=\mathbb{I}.$ Noting that, $i+j=Nd+(i\oplus j)$, for some positive integer $N$, $(\mathbb{X}_d^b\mathbb{Z}_d^a)^i(\mathbb{X}_d^b\mathbb{Z}_d^a)^j=(\mathbb{X}_d^bZ_d^a)^{i\oplus j}$ follows straightforwardly.

\end{appendices}

\end{document}